\documentclass[12pt]{article}
\usepackage[english]{babel}
\usepackage[utf8]{inputenc}
\usepackage{lastpage}
\usepackage{graphicx}
\usepackage{epstopdf}
\usepackage{pgf,tikz}
\usepackage{mathrsfs}
\usepackage{subfigure}
\usepackage{algorithm}
\usepackage{mathtools}
\usepackage{amsthm}
\usepackage{amsmath}
\usepackage{amssymb}
\usepackage{amsfonts}
\usepackage{calrsfs}
\usepackage{caption}
\usepackage{hyperref}
\usepackage{mathrsfs}
\usepackage[autostyle=true]{csquotes}
\usepackage[noend]{algpseudocode}
\usepackage{algorithmicx}
\usepackage{pgfplots}

\usetikzlibrary{arrows}

\title{On the Approximation Ratio of the 3-Opt Algorithm for the (1,2)-TSP}
\author{Xianghui Zhong\\
\small University of Bonn, Germany\\[5mm]
            }
\date{\today} 

\newtheorem{thm}{Theorem}[section]
\theoremstyle{plain}
\newtheorem{lemma}[thm]{Lemma}
\newtheorem{theorem}[thm]{Theorem}
\newtheorem{corollary}[thm]{Corollary}

\theoremstyle{remark}

\theoremstyle{definition}
\newtheorem{definition}[thm]{Definition}

\newcommand*{\N}{\ensuremath{\mathbb{N}}}
\newcommand*{\Z}{\ensuremath{\mathbb{Z}}}

\definecolor{ffqqqq}{rgb}{1.,0.,0.}
\definecolor{uuuuuu}{rgb}{0,0,0}

\begin{document}

\maketitle
\begin{abstract}
The \textsc{(1,2)-TSP} is a special case of the TSP where each edge has cost either 1 or 2. In this paper we give a lower bound of $\frac{3}{2}$ for the approximation ratio of the 2-Opt algorithm for the \textsc{(1,2)-TSP}. Moreover, we show that the 3-Opt algorithm has an exact approximation ratio of $\frac{11}{8}$ for the \textsc{(1,2)-TSP}. Furthermore, we introduce the 3-Opt++-algorithm, an improved version of the 3-Opt algorithm for the \textsc{(1-2)-TSP} with an exact approximation ratio of $\frac{4}{3}$.
\end{abstract}

{\small\textbf{keywords:} traveling salesman problem; \textsc{(1,2)-TSP}; 3-Opt algorithm; approximation algorithm; approximation ratio}

\section{Introduction}
The traveling salesman problem (TSP) is probably the best-known problem in discrete optimization. An instance consists of the pairwise distances of $n$ vertices and the task is to find a shortest Hamiltonian cycle, i.e.\ a tour visiting every vertex exactly once. The problem is known to be NP-hard \cite{TSPNPHARD}. 

In this paper we investigate the \textsc{(1,2)-TSP} which is a special case of the TSP. In this case, the distances are restricted to be 1 or 2. This variant is still NP-hard \cite{TSPNPHARD}. To get good tours in practice several approximation algorithms are considered. They compute in polynomial time a best possible tour. The approximation ratio is one way to compare approximation algorithms. It is the worst-case ratio between the length of the tour computed by the approximation algorithm and that of the optimal tour.

One class of the approximation algorithms are local search algorithms. These algorithms start with an arbitrary tour and make small local changes improving the tour until it is not possible anymore. The most natural local search approximation algorithm is the $k$-Opt algorithm. It starts with an arbitrary tour and replaces at most $k$ edges in every iteration to get a shorter tour. If there is no such improvement anymore, it outputs the tour. 

The currently best approximation ratio for the \textsc({1,2)-TSP} is $\frac{8}{7}$ achieved by two different algorithms by Berman, Karpinski \cite{DBLP:journals/eccc/ECCC-TR05-069} and Adamaszek, Mnich and Paluch \cite{DBLP:conf/icalp/AdamaszekMP18}. The algorithm given by Berman and Karpinski is a local search algorithm based on $k$-Opt. Hence, in contrast to the other TSP variants, for the \textsc{(1,2)-TSP} one of the currently best approximation algorithms in respect of approximation ratio is a local search algorithm. Karpinski and Berman conjectured that their analysis of the algorithm is not tight. Thus, a better understanding of the approximation ratio of local search algorithms could improve the best approximation ratio for the \textsc{(1,2)-TSP}.

The exact approximation ratio of the $k$-Opt algorithm is not known for general $k$ for the \textsc{(1,2)-TSP}. Khanna, Motwani, Sudan and Vazirani gave an upper bound of $\frac{3}{2}$ on the approximation ratio of the 2-Opt algorithm \cite{khanna1998syntactic}. Zhong showed that the appoximation ratio of the $k$-Opt algorithm is at least $\frac{11}{10}$ for any fixed $k$ \cite{zhong2020approximation}. Thus, the exact approximation ratio of the $k$-Opt algorithm is between $\frac{11}{10}$ and $\frac{3}{2}$ for any fixed $k$ which could be better than the currently best approximation ratio of $\frac{8}{7}$ for $k$ large enough. In this paper we make a start by finding the exact approximation ratio for $k=3$. \medskip

\noindent\textbf{New results.}
First, we give a lower bound of $\frac{3}{2}$ on the approximation ratio of the 2-Opt algorithm for the \textsc{(1,2)-TSP}. A matching upper bound of $\frac{3}{2}$ was given in \cite{khanna1998syntactic} and it was noted that this bound can be shown to be tight. Nevertheless, no lower bound was given explicitly. 

Next, we show that the exact approximation ratio of the 3-Opt algorithm is $\frac{11}{8}$ for the \textsc{(1,2)-TSP}.

\begin{theorem} \label{apx ratio of 3-Opt algorithm}
The exact approximation ratio of the 3-Opt algorithm for \textsc{(1,2)-TSP} is $\frac{11}{8}$.
\end{theorem}

The analysis of the 3-Opt algorithm revealed that vertices incident to two edges of cost 2 in the computed tour by the 3-Opt algorithm make the analysis more complicated and increase the approximation ratio. This can be avoided by a small modification. We introduce the $k$-Opt++ algorithm, a slightly modified version of the $k$-Opt algorithm for \textsc{(1,2)-TSP}. We also analyze the exact approximation ratio of the 3-Opt++ algorithm.
\begin{theorem} \label{apx ratio of 3-Opt++ algorithm}
The exact approximation ratio of the 3-Opt++ algorithm for \textsc{(1,2)-TSP} is $\frac{4}{3}$.
\end{theorem}

\noindent\textbf{Outline of the Paper.}
First, we give a lower bound for the approximation ratio of the 2-Opt algorithm of $\frac{3}{2}$ in Section \ref{sec 2 opt lower}. Then, we show in Section \ref{sec 3 opt} that the approximation ratio of the 3-Opt algorithm is $\frac{11}{8}$.
In Section \ref{sec 3 opt ++} we introduce the $k$-Opt++ algorithm, a slightly modified version of the $k$-Opt algorithm, for the \textsc{(1,2)-TSP}. We prove that the approximation ratio of the 3-Opt++ algorithm is $\frac{4}{3}$. 

\section{Lower Bound on the Approximation Ratio of the 2-Opt Algorithm} \label{sec 2 opt lower}
In this section we give a lower bound of $\frac{3}{2}$ on the approximation ratio of the 2-Opt algorithm for \textsc{(1,2)-TSP}. Note that in \cite{khanna1998syntactic} a matching upper bound of $\frac{3}{2}$ was given and it was noted that this bound can be proven to be tight. Nevertheless, an explicit construction for the lower bound was not given. The construction is based on the construction for a lower bound of $2\left(1-\frac{1}{n}\right)$ on the approximation ratio of the $k$-Opt algorithm for \textsc{Metric TSP} in \cite{rosenkrantz1977analysis}. 

We construct an instance with $n$ vertices $\{v_1,\dots,v_n\}$ and a 2-optimal tour $T$. For the instance set the cost of the edges of $\{\{v_i,v_{i+2}\}\vert i\in \{1,\dots,n-2\},i \text{ odd}\} \cup \{\{v_i,v_{i+1} \}\vert i\in \{1,\dots,n-1\}\}$ and $\{v_n,v_1\}$ to 1 and the cost of all other edges to 2. The tour $T$ consists of the edges $\{v_1,v_2\}, \{v_{n-1},v_n\}$ and $\{\{v_i, v_{i+2}\}\vert i\in \{1,\dots, n-2\}\}$ (Figure \ref{construction 12TSP}). 

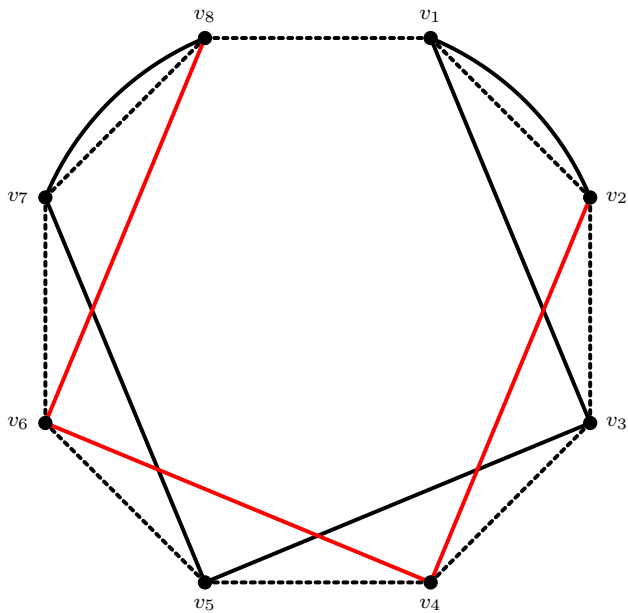
\begin{figure}[!htb]
\centering
 \begin{tikzpicture}[line cap=round,line join=round,>=triangle 45,x=1.5cm,y=1.5cm]
\draw [line width=1.5pt,dotted] (4.,1.8284271247461903)-- (6.,1.8284271247461898);
\draw [line width=1.5pt,dotted] (6.,1.8284271247461898)-- (7.414213562373095,0.4142135623730946);
\draw [line width=1.5pt,dotted] (7.414213562373095,0.4142135623730946)-- (7.414213562373095,-1.585786437626905);
\draw [line width=1.5pt,dotted] (7.414213562373095,-1.585786437626905)-- (6.,-3.);
\draw [line width=1.5pt,dotted] (6.,-3.)-- (4.,-3.);
\draw [line width=1.5pt,dotted] (4.,-3.)-- (2.5857864376269046,-1.5857864376269044);
\draw [line width=1.5pt,dotted] (2.5857864376269046,-1.5857864376269044)-- (2.585786437626905,0.41421356237309537);
\draw [line width=1.5pt,dotted] (2.585786437626905,0.41421356237309537)-- (4.,1.8284271247461903);
\draw [line width=1.5pt,color=uuuuuu] (6.,1.8284271247461898)-- (7.414213562373095,-1.585786437626905);
\draw [line width=1.5pt,color=uuuuuu] (7.414213562373095,-1.585786437626905)-- (4.,-3.);
\draw [line width=1.5pt,color=uuuuuu] (4.,-3.)-- (2.585786437626905,0.41421356237309537);
\draw [line width=1.5pt,color=ffqqqq] (4.,1.8284271247461903)-- (2.5857864376269046,-1.5857864376269044);
\draw [line width=1.5pt,color=ffqqqq] (2.5857864376269046,-1.5857864376269044)-- (6.,-3.);
\draw [line width=1.5pt,color=ffqqqq] (6.,-3.)-- (7.414213562373095,0.4142135623730946);
\draw [shift={(5.,-0.5857864376269051)},line width=1.5pt]  plot[domain=1.963495408493621:2.748893571891069,variable=\t]({1.*2.6131259297527536*cos(\t r)+0.*2.6131259297527536*sin(\t r)},{0.*2.6131259297527536*cos(\t r)+1.*2.6131259297527536*sin(\t r)});
\draw [shift={(5.,-0.5857864376269051)},line width=1.5pt]  plot[domain=0.39269908169872403:1.1780972450961724,variable=\t]({1.*2.6131259297527527*cos(\t r)+0.*2.6131259297527527*sin(\t r)},{0.*2.6131259297527527*cos(\t r)+1.*2.6131259297527527*sin(\t r)});
\begin{scriptsize}
\draw [fill=black] (4.,-3.) circle (2.5pt);
\node[label=below:$v_5$] at (4.,-3.) {};
\draw [fill=black] (6.,-3.) circle (2.5pt);
\node[label=below:$v_4$] at (6.,-3.) {};
\draw [fill=uuuuuu] (7.414213562373095,-1.585786437626905) circle (2.5pt);
\node[label=right:$v_3$] at (7.414213562373095,-1.585786437626905) {};
\draw [fill=uuuuuu] (7.414213562373095,0.4142135623730946) circle (2.5pt);
\node[label=right:$v_2$] at (7.414213562373095,0.4142135623730946) {};
\draw [fill=uuuuuu] (6.,1.8284271247461898) circle (2.5pt);
\node[label=above:$v_1$] at (6.,1.8284271247461898) {};
\draw [fill=uuuuuu] (4.,1.8284271247461903) circle (2.5pt);
\node[label=above:$v_8$] at (4.,1.8284271247461903) {};
\draw [fill=uuuuuu] (2.585786437626905,0.41421356237309537) circle (2.5pt);
\node[label=left:$v_7$] at (2.585786437626905,0.41421356237309537) {};
\draw [fill=uuuuuu] (2.5857864376269046,-1.5857864376269044) circle (2.5pt);
\node[label=left:$v_6$] at (2.5857864376269046,-1.5857864376269044) {};
\end{scriptsize}
\end{tikzpicture}
  \caption{The constructed tour for $n=8$. The black and red edges have cost 1 and 2, respectively. $T$ is the straight tour, the optimal tour is dotted.}
  \label{construction 12TSP}
\end{figure}

The number of edges with cost 2 in $T$ is $\lfloor\frac{n-2}{2} \rfloor$. Thus, $T$ has total length $n+\lfloor\frac{n-2}{2}\rfloor$. The optimal tour for this instance has length $n$ since the tour visiting $v_1,v_2,\dots, v_n$ in this order has length $n$ and is hence optimal. Therefore, the approximation ratio is at least $\lim_{n\to \infty}\frac{n+\lfloor\frac{n-2}{2} \rfloor}{n}=\frac{3}{2}$. It remains to show that $T$ is indeed 2-optimal for large $n$.

\begin{lemma} \label{2Opt 1-2 T optimal}
The tour $T$ constructed above is 2-optimal for $n\geq 7$.
\end{lemma}

\begin{proof}
Assume that there exists an improving 2-move. Then, this 2-move replaces at least an edge of length two. Fix an orientation of $T$ such that the tour edge $(v_2,v_4)$ is oriented this way. It is easy to see that all edges with cost two of the form $(v_i,v_{i+2})$ with $i$ even are oriented this way while for $i$ odd the edges are oriented as $(v_{i+2},v_i)$. 

Assume that the improving 2-move replaces two edges of cost two. Then, these edges have by the definition of $T$ the form $(v_i,v_{i+2})$ and $(v_j,v_{j+2})$ for even $i,j\in\{2,\dots, n-2\}, i\neq j$. According to the fixed orientation the 2-move replaces $(v_i,v_{i+2})$ and $(v_j,v_{j+2})$ by $\{v_i,v_j\}$ and $\{v_{i+2},v_{j+2}\}$. Since $i\neq j$ both even, the edges $\{v_i,v_j\}$ and $\{v_{i+2},v_{j+2}\}$ both have cost two. Thus, this 2-move is not improving contradicting the assumption.

It remains the case that the improving 2-move replaces an edge $(v_i,v_l)$ of cost 1 and an edge $(v_j,v_{j+2})$ of cost 2. Then the new edges $\{v_i,v_j\}$ and $\{v_{l},v_{j+2}\}$ both must have length 1. We distinguish two cases: either $l=i-2$ with $i$ odd or $\{v_i,v_{l}\}=\{v_h,v_{h+1}\}$ for some $h\in\{1, n-1\}$. In the first case the difference of the indices of at least one of the new edges $\{v_i,v_j\}$ or $\{v_{i-2},v_{j+2}\}$ has to be at least three. Hence, at least one new edge has cost 2, contradicting the assumption that we have an improving 2-move. In the second case note that the vertices of the edges $\{v_h,v_{h+1}\}$ and $\{v_j,v_{j+2}\}$ are disjoint, otherwise we do not get a tour after the 2-move. Hence, the difference of the indices of at least one new edge has to be at least three and the total cost of the new edges cannot be 2.
\end{proof}

\begin{theorem}
The approximation ratio of the 2-Opt algorithm for \textsc{(1,2)-TSP} is at least $\frac{3}{2}$.
\end{theorem}

\begin{proof}
We have constructed an instance with a tour $T$ which is 2-optimal by Lemma~\ref{2Opt 1-2 T optimal}. Recall that the length of $T$ is $n+\lfloor \frac{n-2}{2}\rfloor$ while the length of the optimal tour is $n$. Hence, the approximation ratio is at least $\lim_{n\to \infty} \frac{n+\lfloor \frac{n-2}{2} \rfloor}{n}=\frac{3}{2}$.
\end{proof}

\section{Approximation Ratio of the 3-Opt Algorithm} \label{sec 3 opt}
In this section we show that the exact approximation ratio of the 3-Opt algorithm is $\frac{11}{8}$.

\subsection{Lower Bound on the Approximation Ratio of the 3-Opt Algorithm}
We construct for all integer $s\geq 3$ an instance $I_{s}$ with the vertices $\{v_0,\dots, v_{8s-1}\}$ together with a tour $T$. We show that for all even $s\geq 12$ the construced tour $T$ for the instance $I_{s}$ is 3-optimal. Moreover, the ratio of its length and that of the optimal tour is at least $\frac{11s}{8s+4}$.

For simplicity we consider from now on all indicies modulo $8s$ for some fixed $s$. Set the cost of the edges $\{\{v_{8h},v_{8h+1}\}$, $\{v_{8h+1},v_{8h+2}\}$, $\{v_{8h+2}, v_{8h+3}\}$, $\{v_{8h+3},v_{8h+4}\}$, $\{v_{8h+4},v_{8h+5}\}$, $\{v_{8h+2},v_{8h+5}\}$, $\{v_{8h+2},v_{8(h+1)+5}\}$, $\{v_{8h+3},v_{8h}\}$, $\{v_{8h+3},v_{8(h-1)}\}$, \\$\{v_{8h+4},v_{8h+6}\}$, $\{v_{8h+4},v_{8(h+1)+6}\}$, $\{v_{8h+7},v_{8(h+1)+1}\}$, $\{v_{8h+7},v_{8(h+2)+1}\} \vert h\in \Z\}$ to 1 and the cost of all other edges to 2 (Figure \ref{3-Opt lower}).

The tour $T$ consists of the edges $\{\{v_i,v_{i+1}\} \vert i\in \Z \}$. From the construction it is easy to see that the cost of the tour $T$ is $11s$. Next, we bound the length of the optimal tour. 

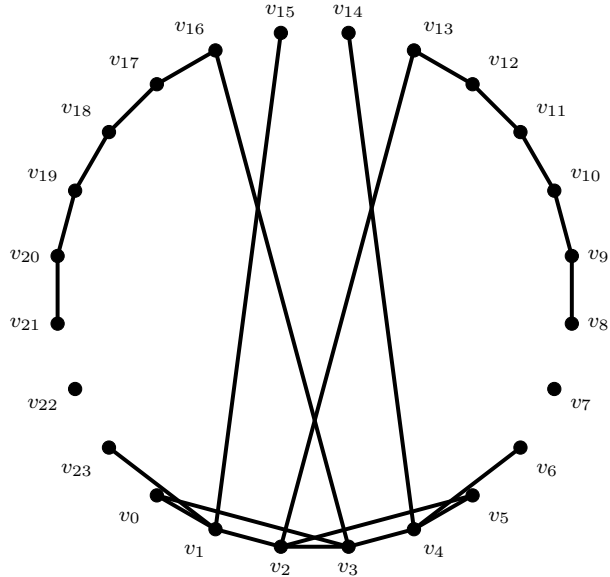
\begin{figure}[!htb]
\centering
\begin{tikzpicture}[line cap=round,line join=round,>=triangle 45,x=0.9cm,y=0.9cm]
\draw [line width=1.5pt] (-0.8919512300735062,1.7788190451025212)-- (-0.025925826289070253,1.2788190451025216);
\draw [line width=1.5pt] (-0.025925826289070253,1.2788190451025216)-- (0.94,1.02);
\draw [line width=1.5pt] (0.94,1.02)-- (1.94,1.02);
\draw [line width=1.5pt] (1.94,1.02)-- (2.9059258262890655,1.2788190451025199);
\draw [line width=1.5pt] (2.9059258262890655,1.2788190451025199)-- (3.7719512300735034,1.7788190451025194);
\draw [line width=1.5pt] (5.237877056362571,4.317877056362571)-- (5.237877056362571,5.317877056362569);
\draw [line width=1.5pt] (5.237877056362571,5.317877056362569)-- (4.979058011260051,6.283802882651636);
\draw [line width=1.5pt] (4.979058011260051,6.283802882651636)-- (4.479058011260051,7.149828286436074);
\draw [line width=1.5pt] (4.479058011260051,7.149828286436074)-- (3.771951230073506,7.85693506762262);
\draw [line width=1.5pt] (3.771951230073506,7.85693506762262)-- (2.905925826289068,8.35693506762262);
\draw [line width=1.5pt] (-0.025925826289064924,8.356935067622622)-- (-0.8919512300735022,7.856935067622624);
\draw [line width=1.5pt] (-0.8919512300735022,7.856935067622624)-- (-1.5990580112600496,7.149828286436078);
\draw [line width=1.5pt] (-1.5990580112600496,7.149828286436078)-- (-2.099058011260049,6.28380288265164);
\draw [line width=1.5pt] (-2.099058011260049,6.28380288265164)-- (-2.357877056362571,5.317877056362572);
\draw [line width=1.5pt] (-2.357877056362571,5.317877056362572)-- (-2.3578770563625713,4.3178770563625735);
\draw [line width=1.5pt] (-0.8919512300735062,1.7788190451025212)-- (1.94,1.02);
\draw [line width=1.5pt] (0.94,1.02)-- (3.7719512300735034,1.7788190451025194);
\draw [line width=1.5pt] (-0.025925826289070253,1.2788190451025216)-- (-1.5990580112600519,2.4859258262890678);
\draw [line width=1.5pt] (-0.025925826289070253,1.2788190451025216)-- (0.94,8.615754112725142);
\draw [line width=1.5pt] (2.9059258262890655,1.2788190451025199)-- (4.479058011260051,2.4859258262890656);
\draw [line width=1.5pt] (2.9059258262890655,1.2788190451025199)-- (1.94,8.615754112725142);
\draw [line width=1.5pt] (0.94,1.02)-- (2.905925826289068,8.35693506762262);
\draw [line width=1.5pt] (1.94,1.02)-- (-0.025925826289064924,8.356935067622622);
\begin{scriptsize}
\draw [fill=black] (0.94,1.02) circle (2.5pt);
\node[label=below:$v_2$] at (0.94,1.02) {};
\draw [fill=black] (1.94,1.02) circle (2.5pt);
\node[label=below:$v_3$] at (1.94,1.02) {};
\draw [fill=uuuuuu] (2.9059258262890655,1.2788190451025199) circle (2.5pt);
\node[label=280:$v_4$] at (2.9059258262890655,1.2788190451025199) {};
\draw [fill=uuuuuu] (3.7719512300735034,1.7788190451025194) circle (2.5pt);
\node[label=315:$v_5$] at (3.7719512300735034,1.7788190451025194) {};
\draw [fill=uuuuuu] (4.479058011260051,2.4859258262890656) circle (2.5pt);
\node[label=315:$v_6$] at (4.479058011260051,2.4859258262890656) {};
\draw [fill=uuuuuu] (4.97905801126005,3.351951230073503) circle (2.5pt);
\node[label=350:$v_7$] at (4.97905801126005,3.351951230073503) {};
\draw [fill=uuuuuu] (5.237877056362571,4.317877056362571) circle (2.5pt);
\node[label=right:$v_8$] at (5.237877056362571,4.317877056362571) {};
\draw [fill=uuuuuu] (5.237877056362571,5.317877056362569) circle (2.5pt);
\node[label=right:$v_{9}$] at (5.237877056362571,5.317877056362569) {};
\draw [fill=uuuuuu] (4.979058011260051,6.283802882651636) circle (2.5pt);
\node[label=10:$v_{10}$] at (4.979058011260051,6.283802882651636) {};
\draw [fill=uuuuuu] (4.479058011260051,7.149828286436074) circle (2.5pt);
\node[label=45:$v_{11}$] at (4.479058011260051,7.149828286436074) {};
\draw [fill=uuuuuu] (3.771951230073506,7.85693506762262) circle (2.5pt);
\node[label=45:$v_{12}$] at (3.771951230073506,7.85693506762262) {};
\draw [fill=uuuuuu] (2.905925826289068,8.35693506762262) circle (2.5pt);
\node[label=80:$v_{13}$] at (2.905925826289068,8.35693506762262) {};
\draw [fill=uuuuuu] (1.94,8.615754112725142) circle (2.5pt);
\node[label=above:$v_{14}$] at (1.94,8.615754112725142) {};
\draw [fill=uuuuuu] (0.94,8.615754112725142) circle (2.5pt);
\node[label=above:$v_{15}$] at (0.94,8.615754112725142) {};
\draw [fill=uuuuuu] (-0.025925826289064924,8.356935067622622) circle (2.5pt);
\node[label=100:$v_{16}$] at (-0.025925826289064924,8.356935067622622) {};
\draw [fill=uuuuuu] (-0.8919512300735022,7.856935067622624) circle (2.5pt);
\node[label=135:$v_{17}$] at (-0.8919512300735022,7.856935067622624) {};
\draw [fill=uuuuuu] (-1.5990580112600496,7.149828286436078) circle (2.5pt);
\node[label=135:$v_{18}$] at (-1.5990580112600496,7.149828286436078) {};
\draw [fill=uuuuuu] (-2.099058011260049,6.28380288265164) circle (2.5pt);
\node[label=170:$v_{19}$] at (-2.099058011260049,6.28380288265164) {};
\draw [fill=uuuuuu] (-2.357877056362571,5.317877056362572) circle (2.5pt);
\node[label=left:$v_{20}$] at (-2.357877056362571,5.317877056362572) {};
\draw [fill=uuuuuu] (-2.3578770563625713,4.3178770563625735) circle (2.5pt);
\node[label=left:$v_{21}$] at (-2.3578770563625713,4.3178770563625735) {};
\draw [fill=uuuuuu] (-2.099058011260052,3.3519512300735075) circle (2.5pt);
\node[label=190:$v_{22}$] at (-2.099058011260052,3.3519512300735075) {};
\draw [fill=uuuuuu] (-1.5990580112600519,2.4859258262890678) circle (2.5pt);
\node[label=225:$v_{23}$] at (-1.5990580112600519,2.4859258262890678) {};
\draw [fill=uuuuuu] (-0.8919512300735062,1.7788190451025212) circle (2.5pt);
\node[label=225:$v_0$] at (-0.8919512300735062,1.7788190451025212) {};
\draw [fill=uuuuuu] (-0.025925826289070253,1.2788190451025216) circle (2.5pt);
\node[label=260:$v_1$] at (-0.025925826289070253,1.2788190451025216) {};
\end{scriptsize}
\end{tikzpicture}
  \caption{The instance $I_3$ where the black edges have cost 1. Due to clarity, not all edges with cost 1 are drawn. The drawn pattern of edges with cost 1 repeats periodically. The tour $T$ connects adjacent vertices on the circle.}
  \label{3-Opt lower}
\end{figure}

\begin{lemma} \label{3 Opt length optimal tour}
The optimal tour $T^*$ for $I_s$ has length at most $8s+4$.
\end{lemma}

\begin{proof}
Note that the four sets 
\begin{align*}
&\{\{v_{8h+2},v_{8h+5}\}, \{v_{8h+2},v_{8(h+1)+5}\}\vert h\in \Z\},\{\{v_{8h+3},v_{8h}\},\{v_{8h+3},v_{8(h-1)}\} \vert h\in \Z\},\\
&\{\{v_{8h+4},v_{8h+6}\}, \{v_{8h+4},v_{8(h+1)+6}\} \vert h\in \Z\}, \{\{v_{8h+7},v_{8(h+1)+1}\}, \{v_{8h+7},v_{8(h+2)+1}\}\vert h\in \Z\}
\end{align*}
form four vertex disjoint cycles with edges of cost 1 whose union visits every vertex exactly once. We can construct a tour of cost at most $8s+4$ for the instance by removing an arbitrary edge from each cycle and connect the four paths arbitrarily to a tour.
\end{proof}

To show the $3$-optimality of $T$ we could make a big case distinction. But instead, we use the next lemma that allows us to perform a computer-assisted proof.

\begin{definition}
A family of instances $(I'_s)_{s\in \N}$ for \textsc{(1,2)-TSP} is called \emph{regular}, if the following conditions are satisfied:
\begin{itemize}
\item There is an $l\in \N$ such that the vertices of $I'_s$ can be labeled by $v_0,\dots, v_{ls-1}$. In the following we consider the indicies modulo $ls$. We partition the vertices in \emph{segments} such that each segment consists of the vertices $\{v_{hl},\dots,v_{(h+1)l-1}\}$ for some $h\in \Z$.
\item The edge $\{v_i,v_j\}$ has cost 1 if and only if $\{v_{i+l},v_{j+l}\}$ has cost 1.
\item If $v_i$ does not lie in the same segment as any of $v_j, v_{j-l}$ and $v_{j+l}$, then the edge $\{v_i,v_j\}$ has cost 2.
\end{itemize}
\end{definition}

\begin{lemma} \label{regular k optimality}
For a regular family of \textsc{(1,2)-TSP} instances $(I'_s)_{s\in \N}$ let $T'_s$ be the tour for $I'_s$ consisting of the edges $\{v_i,v_{i+1}\}$ for $i\in \Z$. We have if $T'_{2k}$ is $k$-optimal then $T'_s$ is also $k$-optimal for all $s\geq 2k$.
\end{lemma}

\begin{proof}
Assume that there is an improving $k$-move for some $T'_s$ with $s\geq 2k$. We show that there is an improving $k$-move for $T'_{2k}$. Assume that the $k$-move removes the edges $e_1,\dots, e_k$ which lie on $T_s$ in this cyclic order. If there is an $i\in \{1,\dots, k\}$ such that between $e_i$ and $e_{i+1}$ lie more than two complete segments where $e_{k+1}:=e_1$, then we can map this $k$-move to $T'_{s-1}$ in $I'_{s-1}$ by removing one of these segments. More precisely, we map $e_1,\dots, e_k$ to $T'_{s-1}$ in $I'_{s-1}$ without changing their positions in the segments and the distances between the edges $e_j$ and $e_{j+1}$ for all $j\in\{1,\dots, k\}\backslash \{i\}$. The distance between $e_i$ and $e_{i+1}$ is by the length of a segment, i.e.\ by $l$, shorter as in $I'_s$. After removing $e_1, \dots, e_k$ the new $k$-move for $T'_{s-1}$ connects the same endpoints of $e_1, \dots, e_k$ as the original $k$-move for $T'_{s}$. By the regular property, the cost of the edges we add and remove by the two $k$-moves are the same. Repeat this procedure and stop if we have a $k$-move for $T'_{2k}$ in $I'_{2k}$. If no such modifications are possible, there is at most one complete segment between $e_j$ and $e_{j+1}$ for all $j\in\{1,\dots, k\}$, therefore the instance has at most $2k$ segments. Hence, in the end we get an improving $k$-move for $T'_{2k}$.
\end{proof}

\begin{lemma} \label{3 Opt T 3-optimal}
The constructed tour $T$ is 3-optimal for $I_s$ with $s$ even and $s\geq 12$.
\end{lemma}

\begin{proof}
For $l=8$ the instances $I_s$ are not regular but satisfy the first two conditions of regularity by construction. The third condition is violated since there are for example edges $\{v_{8h+7},v_{8(h+2)+1}\}$ of cost 1 whose endpoints are up to 2 segments apart. We can construct the instance $I'_s:=I_{2s}$ and choose $l=16$ to get a regular family of instances since the third condition is also satisfied. Therefore, by Lemma \ref{regular k optimality} it is enough to check that the constructed tour $T$ in $I'_6=I_{12}$ is 3-optimal. We checked this using a self-written computer program that generated all possible 3-moves for the instance and observed that none of them is improving. 
\end{proof}

\begin{theorem}
The approximation ratio of the 3-Opt algorithm for \textsc{(1,2)-TSP} is at least $\frac{11}{8}$.
\end{theorem}

\begin{proof}
By Lemma \ref{3 Opt T 3-optimal}, the constructed tour $T$ is 3-optimal for $I_s$ with $s$ even and $s\geq 12$. By construction, it has length $11s$ and by Lemma \ref{3 Opt length optimal tour} the length of the optimal tour is at most $8s+4$. Thus, the approximation ratio is at least $\frac{11s}{8s+4}\to \frac{11}{8}$ for $s\to \infty$.
\end{proof}

\subsection{Upper Bound on the Approximation Ratio of the 3-Opt Algorithm}
For the upper bound assume that an instance with a 3-optimal tour $T$ is given. Let $T^*$ be a fixed optimal tour of the instance. We want to bound the length of $T$ compared to the length of $T^*$ by bounding the number of edges of cost 2. Our general strategy is like \cite{DBLP:journals/eccc/ECCC-TR05-069} to distribute counters to the vertices such that on the one hand if there are many edges of cost 2 in $T$, many counters are distributed. On the other hand, for many counters we need many edges of length 1 to avoid creating an improving 3-move. This way, we get a lower bound on the fraction of the edges with length 1 in $T$ and this implies an upper bound on the approximation ratio. 

We start by describing a procedure to distribute the counters to the vertices of the tour $T$. Then, in the first half of this subsection we show properties of the counters that is ensured by the 3-optimality of the tour. Using these properties we show in the second half that we do not distribute too many counters. In order to show this, we build a linear program whose objective value is an upper bound on the number of counters distributed. We give a solution of the corresponding dual LP to obtain the upper bound. In the end we notice that we distribute many counters if the number of edges of cost 2 is large in $T$. As we have an upper bound on the number of counters we get an upper bound on the number of edges with cost 2 in $T$ and therefore on the approximation ratio.

Let the \emph{1-paths} be the connected components we obtain after deleting all edges with cost 2 in $T$. We call the vertices with degree 1 in a 1-path the \emph{endpoints} of the 1-path. Now, we want to distribute counters to the vertices of the tour $T$.

\begin{definition} \label{counter distribution}
We distribute counters as follows: For 1-paths of length 0 consisting of the vertex $v$ we distribute two counters to the vertex $w$ if $\{v,w\}\in T^*, c(v,w)=1$. We call these counters \emph{good}. For every 1-path of length greater than 0 we distribute a counter on $w$ if $v$ is an endpoint of the 1-path and $\{v,w\}\in T^*, c(v,w)=1$. These counters are called \emph{bad}.
\end{definition}

Next, we show some properties of the counters and $T$ we need for the analysis.

\begin{lemma} \label{3Opt 1-edge connecting endpoints}
Let $p,q$ be the endpoints of different 1-paths of a 3-optimal tour $T$, then $c(p,q)=2$.
\end{lemma}

\begin{proof}
Assume there is an edge $\{p,q\}$ connecting two endpoints $p,q$ of different 1-paths with $c(p,q)=1$ we show that there is an improving 3-move. Let $\{p,q\}$ be incident to the two edges $\{p,u\}$ and $\{q,v\}$ with cost 2 in $T$. We perform first a 2-move replacing the edges $\{p,u\}$ and $\{q,v\}$ by $\{p,q\}$ and $\{u,v\}$. Then, the cost decreases since $c(p,u)+c(q,v)=4>1+2\geq c(p,q)+c(u,v)$. Hence, if afterward the tour stays connected we have found an improving 2-move. It remains the case that the tour splits into two connected components. Since $\{p,q\}$ does not connect two endpoints of a single 1-path, there has to be an edge $\{a,b\}$ of cost 2 in the connected component containing $p$ and $q$ (Figure \ref{3Opt Terminal Moves}). We perform a 2-move replacing $\{u,v\}$ and $\{a,b\}$ by $\{a,u\}$ and $\{b,v\}$ to get a connected tour again. Note that in total we performed a single 3-move since we added $\{u,v\}$ and removed it again. In the end the total cost decreased compared to the initial tour $T$ since 
\begin{align*}
c(p,u)+c(q,v)+c(a,b)&=2+2+2>1+2+2\geq c(p,q)+c(a,u)+c(b,v).
\end{align*}
This is a contradiction to the 3-optimality of $T$.
\end{proof}

\begin{figure}[!htb]
\centering
\begin{tikzpicture}[line cap=round,line join=round,>=triangle 45,x=1.3cm,y=1.3cm]
\draw [line width=1.5pt,color=ffqqqq] (-10,2)-- (-8,2);
\draw [line width=1.5pt,color=ffqqqq] (-10.879385241571816,6.987241532966372)-- (-11.879385241571816,5.255190725397496);
\draw [line width=1.5pt,color=black] (-7.120614758428183,6.9872415329663715)-- (-6.120614758428183,5.255190725397494);
\draw [shift={(-9,4.747477419454622)},line width=1.5pt]  plot[domain=0.8726646259971647:2.268928027592628,variable=\t]({1*2.923804400163087*cos(\t r)+0*2.923804400163087*sin(\t r)},{0*2.923804400163087*cos(\t r)+1*2.923804400163087*sin(\t r)});
\draw [shift={(-9,4.747477419454622)},line width=1.5pt]  plot[domain=2.96705972839036:4.363323129985823,variable=\t]({1*2.9238044001630867*cos(\t r)+0*2.9238044001630867*sin(\t r)},{0*2.9238044001630867*cos(\t r)+1*2.9238044001630867*sin(\t r)});
\draw [shift={(-9,4.747477419454622)},line width=1.5pt]  plot[domain=-1.221730476396031:0.17453292519943273,variable=\t]({1*2.9238044001630867*cos(\t r)+0*2.9238044001630867*sin(\t r)},{0*2.9238044001630867*cos(\t r)+1*2.9238044001630867*sin(\t r)});
\draw [line width=1.5pt,dotted] (-10.879385241571816,6.987241532966372)-- (-8,2);
\draw [line width=1.5pt,dotted,color=black] (-11.879385241571816,5.255190725397496)-- (-7.120614758428183,6.9872415329663715);
\draw [line width=1.5pt,dotted,color=black] (-10,2)-- (-6.120614758428183,5.255190725397494);
\begin{scriptsize}
\draw [fill=black] (-10,2) circle (2.5pt);
\draw[color=black] (-9.990377602610709,1.7786739925234968) node {$v$};
\draw [fill=black] (-8,2) circle (2.5pt);
\draw[color=black] (-7.948550281308474,1.7786739925234968) node {$q$};
\draw [fill=black] (-6.120614758428183,5.255190725397494) circle (2.5pt);
\draw[color=black] (-6.037480374670915,5.4694737644688045) node {$b$};
\draw [fill=black] (-7.120614758428183,6.9872415329663715) circle (2.5pt);
\draw[color=black] (-7.043306641322262,7.199494943109121) node {$a$};
\draw [fill=black] (-10.879385241571816,6.987241532966372) circle (2.5pt);
\draw[color=black] (-11.026378657261596,7.129494943109121) node {$p$};
\draw [fill=black] (-11.879385241571816,5.255190725397496) circle (2.5pt);
\draw[color=black] (-11.99197187324689,5.409532027135319) node {$u$};
\end{scriptsize}
\end{tikzpicture}
  \caption{Sketch for Lemma \ref{3Opt 1-edge connecting endpoints} and Lemma \ref{3Opt forbidden constellation}. The tour $T$ consists of the solid edges and the cost of each red edge is 2. In Lemma \ref{3Opt 1-edge connecting endpoints} we have in addition $c(a,b)=2$ and $c(p,q)=1$, while in Lemma \ref{3Opt forbidden constellation} we have instead $c(a,u)=c(b,v)=1$. In both cases we replace the edges $\{p,u\}$, $\{q,v\}$ and $\{a,b\}$ by $\{p,q\}$, $\{a,u\}$ and $\{b,v\}$.}
  \label{3Opt Terminal Moves}
\end{figure}
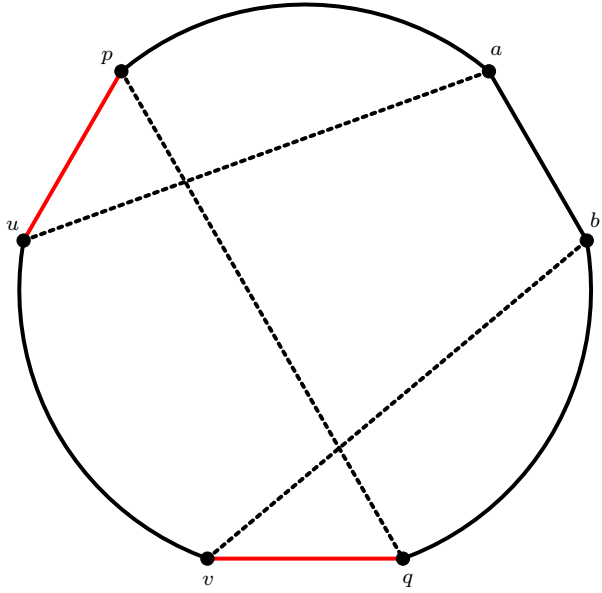

\begin{corollary} \label{3 Opt at most one counter endpoint}
The endpoints $p$ and $q$ of a 1-path of $T$ can only have counters distributed by the edge $\{p,q\}$. In particular, each of them can have at most one bad counter.
\end{corollary}

\begin{proof}
For any other endpoint $r$ by Lemma \ref{3Opt 1-edge connecting endpoints} we have $c(p,r)=c(q,r)=2$. Therefore, these edges cannot assign a counter to $p$ or $q$.
\end{proof}

\begin{lemma} \label{3Opt forbidden constellation}
There are no vertices $p,q, u,v, a,b$ such that $\{p,u\},\{q,v\}, \{a,b\}\in T$, $c(p,u)=c(q,v)=2$, $c(a,u)=c(b,v)=1$ and $p,q,a,b$ lie on the same side of $\{u,v\}$ (Figure \ref{3Opt Terminal Moves}).
\end{lemma}

\begin{proof}
Like in the proof of Lemma \ref{3Opt 1-edge connecting endpoints} we can replace the edges $\{p,u\}$, $\{q,v\}$ and $\{a,b\}$ by $\{p,q\}$, $\{a,u\}$ and $\{b,v\}$. The cost of the tour decreases since in this case we have
\begin{align*}
c(p,u)+c(q,v)+c(a,b)\geq 2+2+1>2+1+1\geq c(p,q)+c(a,u)+c(b,v).
\end{align*}
\end{proof}

\begin{lemma} \label{3 Opt neighbors}
If the vertices $r$ and $t$ have good counters, then $\{r,t\} \not\in T$. Moreover, if there is a vertex $s$ with $\{r,s\},\{s,t\}\in T$, then $s$ does not have a counter.
\end{lemma}

\begin{proof}
For the first statement assume the contrary, then there are two 1-paths of length 0 consisting of the vertices $u$ and $v$ such that $c(r,u)=c(t,v)=1$, respectively. Since $u$ and $v$ are 1-paths of length 0 we can choose $a=r$, $b=t$ and appropriate neighbors $p$ and $q$ to contradict Lemma \ref{3Opt forbidden constellation}.

Similarly, for the second statement assume there are such vertices $r,s,t$. Then, there is an endpoint $w$ of a 1-path with $c(w,s)=1$ and a vertex $z$ with $\{z,w\}\in T$ and $c(z,w)=2$. Now, $z$ lies either on the same side of $\{w,s\}$ as $r$ or $t$. Depending on this we get a contradiction to Lemma \ref{3Opt forbidden constellation} for $a=s,b=t$ or $a=r, b=s$. 
\end{proof}

\begin{lemma} \label{3 Opt endpoints no counter}
Let $p,q$ be the endpoints of a 1-path containing $w$ with $c(w,q)=1$. If $w$ has a good counter, then $p,q$ do not have counters.
\end{lemma}

\begin{proof}
Assume the contrary, by Corollary \ref{3 Opt at most one counter endpoint} the counters of $p$ and $q$ are distributed by the edge $\{p,q\}$ and we must have $c(p,q)=1$. Let the good counter of $w$ be originated from the 1-path of length 0 consisting of the vertex $u$. Moreover, let w.l.o.g.\ $v_1$ be the vertex adjacent to $u$ in $T$ with $c(v_1,u)=2$ and lying on the same side of $\{u,w\}$ as $p$ and let $r$ be the vertex with $\{q,r\}\in T$ and $c(q,r)=2$ (Figure \ref{3-Opt++ Moves}). Then, we can replace the edges $\{v_1,u\}$, $\{p,w\}$ and $\{q,r\}$ by $\{v_1,r\}$, $\{u,w\}$ and $\{p,q\}$. The cost of $T$ decreases since 
\begin{align*}
c(v_1,u)+c(p,w)+c(q,r)&=2+1+2>2+1+1\geq c(v_1,r)+ c(u,w)+c(p,q).
\end{align*}
This is a contradiction to the 3-optimality.
\end{proof}

\begin{figure}[!htb]
\centering
 \begin{tikzpicture}[line cap=round,line join=round,>=triangle 45,x=2cm,y=2cm]
\draw [line width=1.5pt,dotted] (-8.84125353283118,3.5406408174555972)-- (-7,3);
\draw [line width=1.5pt,dotted] (-8,3)-- (-7.5,6.477576385886735);
\draw [line width=1.5pt,dotted] (-6.540507026385503,6.195843829045305)-- (-8.84125353283118,3.5406408174555972);
\draw [line width=1.5pt,dotted] (-8.459492973614497,6.195843829045305)-- (-6.15874646716882,3.5406408174555972);
\draw [line width=1.5pt, color=red] (-8.459492973614497,6.195843829045305)-- (-7.5,6.477576385886735);
\draw [line width=1.5pt, color=red] (-7.5,6.477576385886735)-- (-6.540507026385503,6.195843829045305);
\draw [line width=1.5pt] (-8.84125353283118,3.5406408174555972)-- (-8,3);
\draw [line width=1.5pt] (-8,3)-- (-7,3);
\draw [line width=1.5pt,color=red] (-7,3)-- (-6.15874646716882,3.5406408174555972);
\draw [shift={(-7.5,4.702843619444623)},line width=1.5pt]  plot[domain=2.141994991083949:3.8555909839511098,variable=\t]({1*1.7747327664421109*cos(\t r)+0*1.7747327664421109*sin(\t r)},{0*1.7747327664421109*cos(\t r)+1*1.7747327664421109*sin(\t r)});
\draw [shift={(-7.5,4.702843619444623)},line width=1.5pt]  plot[domain=-0.7139983303613162:0.9995976625058433,variable=\t]({1*1.7747327664421106*cos(\t r)+0*1.7747327664421106*sin(\t r)},{0*1.7747327664421106*cos(\t r)+1*1.7747327664421106*sin(\t r)});
\begin{scriptsize}
\draw [fill=black] (-8,3) circle (2.5pt);
\draw[color=black] (-7.969063201019697,2.884441718890698) node {$w$};
\draw [fill=black] (-7,3) circle (2.5pt);
\draw[color=black] (-6.9406239598445305,2.8916842487581287) node {$q$};
\draw [fill=black] (-6.15874646716882,3.5406408174555972) circle (2.5pt);
\draw[color=black] (-6.078762905620272,3.456601578417726) node {$r$};
\draw [fill=black] (-6.540507026385503,6.195843829045305) circle (2.5pt);
\draw[color=black] (-6.462616988594102,6.324643405920295) node {$v_2$};
\draw [fill=black] (-7.5,6.477576385886735) circle (2.5pt);
\draw[color=black] (-7.440358520697252,6.585374481147802) node {$u$};
\draw [fill=black] (-8.459492973614497,6.195843829045305) circle (2.5pt);
\draw[color=black] (-8.454312702137557,6.346613525390018) node {$v_1$};
\draw [fill=black] (-8.84125353283118,3.5406408174555972) circle (2.5pt);
\draw[color=black] (-8.932319673387987,3.434873988815434) node {$p$};
\end{scriptsize}
\end{tikzpicture}
  \caption{Sketch for Lemma \ref{3 Opt endpoints no counter} and Lemma \ref{3Opt Single}. The solid edges are the edges of the tour $T$ and the red edges have cost 2. In Lemma \ref{3 Opt endpoints no counter} we replace the edges $\{v_1,u\}$, $\{p,w\}$ and $\{q,r\}$ by $\{v_1,r\}$, $\{u,w\}$ and $\{p,q\}$. In Lemma~\ref{3Opt Single} the edges $\{u,v_2\}$ and $\{p,w\}$ are replaced by $\{v_2,p\}$ and $\{u,w\}$. }
  \label{3-Opt++ Moves}
\end{figure}
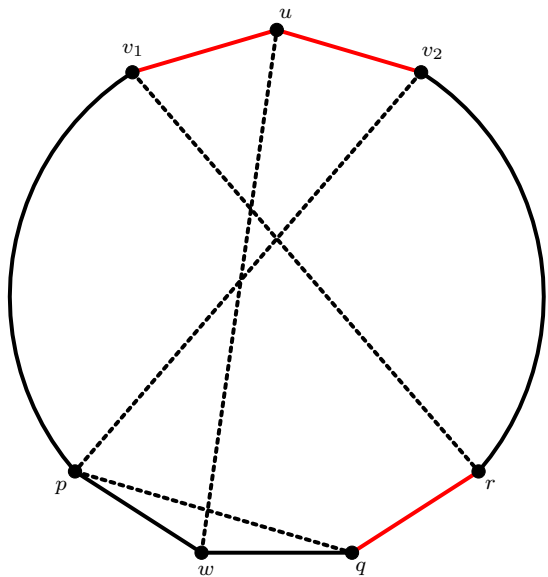

Before we bound the number of counters, we summarize the properties we showed about counters.

\begin{corollary} \label{counters conditions}
The following properties hold:
\begin{enumerate}
\item Every vertex has two \emph{slots} where a slot can be empty or contain two good counters or one bad counter.
\item If the vertices $a$ and $c$ both have good counters, then $\{a,c\}\not\in T$. Moreover, if there is a vertex $b$ such that $\{a,b\},\{b,c\}\in T$, then $b$ does not have counters. \label{prop distance good counters}
\item The endpoint of a 1-path of length 0 does not have counters. Each endpoint of all other 1-paths can only have at most one bad counter and no good counters. \label{prop endpoint bad counters} 
\item If the endpoint $p$ of a 1-path has a counter, then $w$ does not have a good counter if $\{w,p\}\in T$ and $c(w,p)=1$. \label{prop bound good counters}
\item The total number of bad counters is less or equal to four times the number of 1-paths of length greater than 0. \label{prop bad counters}
\end{enumerate}
\end{corollary}

\begin{proof}
The first property is due to the fact that every vertex is incident to at most two edges of cost 1 in the optimal tour and every such edge can distribute two good counters or one bad counter. The second, third and fourth property follow from Lemma \ref{3 Opt neighbors}, Corollary \ref{3 Opt at most one counter endpoint} and Lemma \ref{3 Opt endpoints no counter}, respectively. The fifth property follows from the fact that the bad counters are distributed by the endpoints of 1-paths with length greater than 0 and each such endpoint distributes at most two bad counters. 
\end{proof}

Next, we do not restrict to the distribution of counters according to Definition \ref{counter distribution}. Instead, we investigate arbitrary distributions of good and bad counters to the vertices of $T$. We call an arbitrary distribution of counters \emph{nice} if it satisfies the properties of Corollary \ref{counters conditions}. Moreover, we call an arbitrary distribution of counters \emph{crowded} if there are more than $\frac{12}{5}h$ counters, where $h$ is the number of edges with cost 1 in $T$. By Corollary \ref{counters conditions} we already know that the distribution of counters according to Definition \ref{counter distribution} is nice. We want to show that it is not crowded. To achieve this we show that every nice distribution of counters is not crowded. 

\begin{lemma} \label{distance of good counters}
If there is a nice and crowded distribution of counters, then there is a nice and crowded distribution of counters for possibly a different tour such that: Every vertex with good counters has four good counters and the distance of any two vertices on the same 1-path each having good counters is at least 3.
\end{lemma}

\begin{proof}
Let a nice and crowded distribution of counters for a tour $T$ with $h$ minimal be given, where $h$ is the number of edges with cost 1 in $T$. First, note that if a vertex $v$ has two good counters, then we may assume that it has four good counters instead of exactly two good counters or two good counters and one bad counter. It is easy to check that this assumption does not decrease the number of counters or contradict Corollary~\ref{counters conditions} since $v$ already had a good counter.

Next, we exclude the case that there are two vertices $a$ and $c$ on the same 1-path with good counters such there is a vertex $b$ with $\{a,b\},\{b,c\}\in T$. In this case by property~\ref{prop distance good counters} of Corollary \ref{counters conditions} the vertex $b$ does not have a counter. Let $u$ be the neighbor of $a$ in $T$ other than $b$. We remove the vertices $a$ and $b$ from the instance, remove the edges $\{u,a\}$, $\{a,b\}$ and $\{b,c\}$ and add the edge $\{u,c\}$ to $T$ and set its cost to 1. By property~\ref{prop endpoint bad counters} in Corollary \ref{counters conditions}, $u$ belongs to the same 1-path as $a$ in the original tour, hence the new tour does not contradict property \ref{prop endpoint bad counters}. This also implies that the number of edges with cost 1 decreased by 2 while the number of counters decreased by 4. Moreover, we did not introduce new bad counters or new vertices neighboring to a vertex with a good counter. Hence, the new tour does not contradict Corollary \ref{counters conditions}. Let the number of counters before the modification be $s$. Since the edges $\{u,a\}$, $\{a,b\}$ and $\{b,c\}$ all had cost 1 in the old tour we have $h-2\geq 1$. We conclude that the ratio $\frac{s}{h}$ does not decrease as
\begin{align*}
\frac{s}{h}>\frac{12}{5}>2 \Rightarrow 2s> 4h \Rightarrow \frac{s-4}{h-2}> \frac{s}{h}> \frac{12}{5}.
\end{align*}
This contradicts that we chose an nice and crowded distribution with $h$ minimal. Thus, the distance of two vertices $a$ and $c$ on the same 1-path each having good counters is at least 3.
\end{proof}

\begin{lemma} \label{number of good and bad counters}
If there is a nice and crowded distribution of counters, then there is a nice and crowded distribution of counters such that each 1-path with $i$ edges contains $g_i$ good and at most $b_i$ bad counters, where $b_0=g_0=0$ and for $i>0$
\begin{align*}
b_i&:=4\cdot\frac{i}{3}-2, & g_i&:= 4\cdot\frac{i}{3} & &\text{for } i\equiv 0 \mod 3\\
b_i&:=4\cdot\frac{i-1}{3}+2, & g_i&:= 4\cdot\frac{i-1}{3} & &\text{for } i\equiv 1 \mod 3\\
b_i&:=4\cdot\frac{i-2}{3}, & g_i&:=4\cdot\frac{i+1}{3} & &\text{for } i\equiv 2 \mod 3.
\end{align*}
\end{lemma}
\begin{proof}
We construct nice distributions of counters satisfying the properties in Lemma~\ref{distance of good counters} and having maximal number of counters. One such distribution of counters has to be crowded if there is a nice and crowded distribution of counters. Every 1-path of length 0 does not contain counters by property \ref{prop endpoint bad counters} in Corollary~\ref{counters conditions}. Therefore, it is enough to restrict to the 1-paths with $i>0$ edges. By Lemma \ref{distance of good counters} we may assume that the distance of two vertices on a 1-path each having good counters is at least 3. Since the endpoints of the 1-path do not contain good counters, the number of good counters on a 1-path with $i$ edges is at most $g_i$. To maximize the number of counters, we have to maximize the number of vertices with good counters, since every vertex can contain at most four good counters or two bad counters. For $i\equiv 0 \mod 3$ and $i\equiv 2 \mod 3$ we can either have no counters on the endpoints of the 1-path and $g_i$ good counters and at most $b_i$ bad counters or one bad counter on each of the endpoints but by property \ref{prop bound good counters} in Corollary \ref{counters conditions} only $g_i-4$ good counters and at most $b_i+2$ additional bad counters. The total number of counters is the same but the first case has the advantage that we have fewer bad counters and the total number of bad counters is bounded by property \ref{prop bad counters} in Corollary \ref{counters conditions}. For $i\equiv 1 \mod 3$ each endpoint of the 1-path can contain a bad counter such that in total we have $g_i$ good counters and at most $b_i$ bad counters. Therefore, we can assume that every 1-path with $i$ edges has $g_i$ good counters and at most $b_i$ bad counters. 
\end{proof}

Now, we show that every nice distribution of counters is not crowded. Since by Lemma~\ref{counters conditions} the distribution of counters from Definition \ref{counter distribution} is nice, we conclude that it is not crowded.

\begin{lemma} \label{3 Opt counters 1-path}
Every nice distribution of counters is not crowded. In particular, if we distribute the counters according to Definition \ref{counter distribution}, every 3-optimal tour $T$ has at most $\frac{12}{5}h$ counters, where $h$ is the number of edges with cost 1 in $T$.
\end{lemma}

\begin{proof}
Assume that there is a nice and crowded distribution of counters. By Lemma~\ref{number of good and bad counters} we can choose one such that the 1-paths with $i$ edges have $g_i$ good counters and at most $b_i$ bad counters. Consider the following LP:
\begin{align*}
\max \ \sum_{i\in \N_{>0}}g_ix_i+&z \\
s.t. \qquad \sum_{i\in \N_{>0}} ix_i = &h \\
z- \sum_{i\in \N_{>0}} b_ix_i\leq &0\\
z-4\sum_{i\in \N_{>0}} x_i \leq &0\\
x_i \geq &0 \qquad \forall i\in \N_{>0}
\end{align*}

This LP gives an upper bound on the number of counters in a tour with $h$ edges of cost 1. The variable $x_i$ counts the number of 1-paths with $i$ edges. Since $b_0=g_0=0$, it is enough to restrict to the case $i>0$. The variable $z$ is equal to the total number of bad counters. The first equation ensures that the total number of edges with cost 1 is $h$ while the second inequality bounds the number of bad counters by the sum of the upper bounds of bad counters on each 1-path. The third inequality ensures that the number of bad counters is at most four times the number of 1-paths with length greater than 0 by property \ref{prop bad counters} of Corollary \ref{counters conditions}. 

Note that by the first equation for any fixed $h$ only a finite number of $x_i$ are nonzero. Now, consider the dual LP:
\begin{align*}
\min \ y_1h& \\
s.t. \qquad iy_1-b_iy_2-4y_3 \geq &g_i \qquad  \forall i\in \N_{>0} \\
y_2+y_3 \geq &1 \\
y_2,y_3 \geq &0.
\end{align*}
We show that $y_1=\frac{12}{5}, y_2=\frac{4}{5}$ and $y_3=\frac{1}{5}$ is a feasible dual solution with value $\frac{12}{5}h$. Therefore, by weak duality we get $\frac{12}{5}h$ as upper bound on the primal LP. Obviously, we have $y_2+y_3=1$ and $y_2,y_3\geq 0$. Moreover, we can check the first inequality:
\begin{align*}
i\cdot\frac{12}{5}-\left(4\cdot\frac{i}{3}-2\right)\cdot\frac{4}{5}-4\cdot\frac{1}{5}=4\cdot\frac{i}{3}+\frac{4}{5}&>4\cdot\frac{i}{3} && &\text{for } i\equiv 0 \mod 3\ \\
i\cdot\frac{12}{5}-\left(4\cdot\frac{i-1}{3}+2\right)\cdot\frac{4}{5}-4\cdot\frac{1}{5}&=4\cdot\frac{i-1}{3} && &\text{for } i \equiv 1 \mod 3\ \\
i\cdot\frac{12}{5}-4\cdot\frac{i-2}{3}\cdot\frac{4}{5}-4\cdot\frac{1}{5}&=4\cdot\frac{i+1}{3} && &\text{for } i \equiv 2 \mod 3.
\end{align*}
Thus, the first inequality is also satisfied and we get the upper bound of $\frac{12}{5}h$ on the number of counters. This is a contradiction to the assumption that the distribution of counters is crowded.
\end{proof}

\begin{lemma} \label{3 Opt ratio bound}
Consider the distribution of counters according to Definition \ref{counter distribution}. Let the number of counters in the tour $T$ be at most $d\cdot h$ where $d$ is a constant and $h$ is the number of edges with cost 1 in $T$. Then, the ratio between the length of $T$ and that of the optimal tour is at most $1+\frac{d}{4+d}$.
\end{lemma}

\begin{proof}
Let $l$ and $f$ be the number of edges with cost 2 in $T$ and the optimal tour, respectively. If the optimal tour consists only of edges with cost 1, every 1-path distributes four counters: For every 1-path of length 0 the unique endpoint is incident to two edges of length 1 in the optimal tour and distributes four good counters. Every endpoint of the other 1-paths is incident to two edges of length 1 in the optimal tour and distributes two bad counters. Every edge of cost 2 in the optimal tour decreases the number of counters distributed by at most 4. Note that the number of 1-paths in $T$ is equal to the number of edges with cost 2 which is $l$. Therefore, the $l$ 1-paths distribute at least $4l-4f$ counters and we conclude $d h\geq 4l-4f$ or $h\geq \frac{1}{d}4(l-f)$. Note that the length of the tour $T$ is $h+2l$, while that of the optimal tour $T^*$ is $n+f=h+l+f$. Therefore, we get the upper bound on the ratio of
\begin{align*}
\frac{c(T)}{c(T^*)}&=\frac{h+2l}{h+l+f}=1+\frac{l-f}{h+l+f}\leq 1+ \frac{l-f}{\frac{1}{d}4(l-f)+l+f}= 1+ \frac{l-f}{\frac{4+d}{d}l-\frac{4-d}{d}f}\\
&\leq 1+ \frac{l-f}{\frac{4+d}{d}l-\frac{4+d}{d}f}=1+\frac{d}{4+d}.
\end{align*}
\end{proof}

\begin{theorem}
The approximation ratio of the 3-Opt algorithm for \textsc{(1,2)-TSP} is at most $\frac{11}{8}$.
\end{theorem}

\begin{proof}
By Lemma \ref{3 Opt counters 1-path}, there are at most $\frac{12}{5}h$ counters in $T$ where $h$ is the number of edges with cost 1 in $T$. By Lemma \ref{3 Opt ratio bound}, the approximation ratio is at most $1+\frac{\frac{12}{5}}{4+\frac{12}{5}}=\frac{11}{8}$. 
\end{proof}

\section{Approximation Ratio of the 3-Opt++ Algorithm} \label{sec 3 opt ++}
In the last section we showed that the approximation ratio of the 3-Opt algorithm is $\frac{11}{8}$. The 1-paths with length 0 play a central role in the analysis. Every endpoint of such a 1-path distribute at most four good counters instead of two bad counters and increase the approximation ratio. We introduce the $k$-Opt++ algorithm by adapting a concept from the $k$-improv algorithm from \cite{DBLP:journals/eccc/ECCC-TR05-069}. The $k$-Opt++ algorithm does not only search for improving $k$-moves but also $k$-moves that decrease the number of 1-paths of length 0 without increasing the cost of the tour. We show that the 3-Opt++ algorithm has a better approximation ratio of $\frac{4}{3}$ than 3-Opt.

\subsection{The $k$-Opt++ algorithm}
We first describe the $k$-Opt++ algorithm. An \emph{improving $k$-Opt++-move} is a $k$-move that is either improving or does not change the cost of the tour but decreases the number of 1-paths of length 0 in the tour. A tour is called \emph{k-Opt++-optimal} if there does not exist an improving $k$-Opt++-move. Like the $k$-Opt algorithm the $k$-Opt++ algorithm starts with an arbitrary tour $T$ and performs improving $k$-Opt++ moves until $T$ is $k$-Opt++-optimal (Algorithm \ref{k-Opt++ algorithm}).

\begin{algorithm}
\caption{$k$-Opt++ Algorithm}
\label{k-Opt++ algorithm}
 \hspace*{\algorithmicindent} \textbf{Input:} Instance of \textsc{TSP} $(K_n,c)$ \\
 \hspace*{\algorithmicindent} \textbf{Output:} Tour $T$
\begin{algorithmic}[1]
\State Start with an arbitrary tour $T$
\While{$\exists$ improving $k$-Opt++-move for $T$}  
\State Perform an improving $k$-Opt++-move on $T$
\EndWhile
\State \Return $T$
\end{algorithmic}
\end{algorithm}

Obviously, every $k$-Opt++-optimal tour is also $k$-optimal. Moreover, every iteration of the $k$-Opt++ algorithm can also be performed in polynomial time since the number of $k$-moves is polynomially bounded. Thus, the algorithm also terminates in polynomial time since the length of the initial tour is bounded by $2n$ and every step improves the length of the tour by at least 1 or decreases the number of 1-paths with length 0 by 1.

\subsection{Lower Bound on the Approximation Ratio of the 3-Opt++ Algorithm}
For every natural number $s\geq 2$ we construct an instance $I_s$ with the vertices $\{v_0,\dots, v_{6s-1}\}$ and approximation ratio $\frac{4}{3}$. For simplicity we consider the indicies modulo $6s$. Set the cost of the edges $\{\{v_{6h},v_{6h+1}\}$, $\{v_{6h+2},v_{6h+3}\}$, $\{v_{6h+3}, v_{6h+4}\}$, $\{v_{6h+4},v_{6h+5}\}$, $\{v_{6h},v_{6h+3}\}$, $\{v_{6h+2},v_{6h+5}\}, \{v_{6h+4},v_{6(h+1)+1}\} \vert h\in \Z\}$ to 1 and the cost of all other edges to 2.

The tour $T$ consists of the edges $\{\{v_i,v_{i+1}\} \vert i\in \Z \}$. The optimal tour $T^*$ consists of the edges $\{\{v_{2h+1},v_{2h}\}, \{v_{2h},v_{2h+3}\} \vert h\in \Z \}$ (Figure \ref{12TSP 3 Opt lower}).
It is easy to check that $T^*$ consists only of edges with cost 1 and is hence optimal.

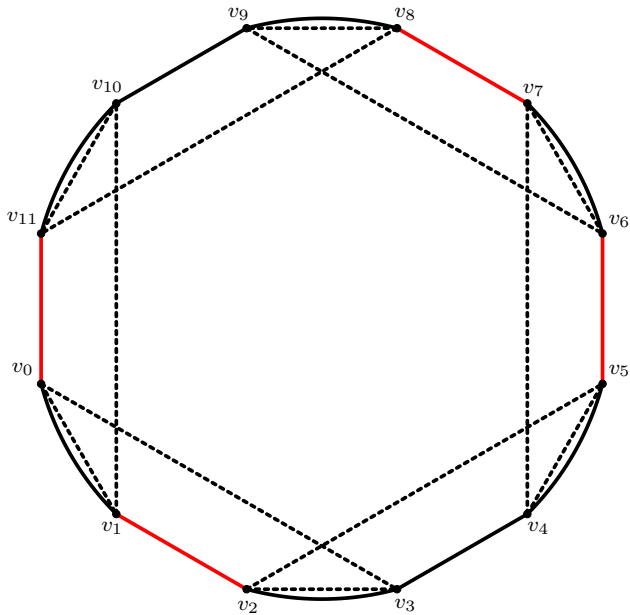
\begin{figure}[!htb]
\centering
\begin{tikzpicture}[line cap=round,line join=round,>=triangle 45,x=2.0cm,y=2.0cm]
\draw [line width=1.4pt,color=ffqqqq] (-1.3660254037844388,2.3660254037844393)-- (-1.3660254037844393,1.3660254037844402);
\draw [line width=1.4pt,color=ffqqqq] (2.366025403784439,2.3660254037844384)-- (2.366025403784439,1.3660254037844388);
\draw [line width=1.4pt,dotted] (-1.3660254037844393,1.3660254037844402)-- (-0.8660254037844393,0.5);
\draw [line width=1.4pt,color=ffqqqq] (-0.8660254037844393,0.5)-- (0,0);
\draw [line width=1.4pt,dotted] (0,0)-- (1,0);
\draw [line width=1.4pt] (1,0)-- (1.8660254037844386,0.5);
\draw [line width=1.4pt,dotted] (1.8660254037844386,0.5)-- (2.366025403784439,1.3660254037844388);
\draw [line width=1.4pt,dotted] (2.366025403784439,2.3660254037844384)-- (1.8660254037844393,3.232050807568877);
\draw [line width=1.4pt,color=ffqqqq] (1.8660254037844393,3.232050807568877)-- (1,3.7320508075688776);
\draw [line width=1.4pt,dotted] (1,3.7320508075688776)-- (0,3.7320508075688776);
\draw [line width=1.4pt] (0,3.7320508075688776)-- (-0.8660254037844384,3.2320508075688785);
\draw [line width=1.4pt,dotted] (-0.8660254037844384,3.2320508075688785)-- (-1.3660254037844388,2.3660254037844393);
\draw [line width=1.4pt,dotted] (-1.3660254037844393,1.3660254037844402)-- (1,0);
\draw [line width=1.4pt,dotted] (0,0)-- (2.366025403784439,1.3660254037844388);
\draw [line width=1.4pt,dotted] (0,3.7320508075688776)-- (2.366025403784439,2.3660254037844384);
\draw [line width=1.4pt,dotted] (1,3.7320508075688776)-- (-1.3660254037844388,2.3660254037844393);
\draw [line width=1.4pt,dotted] (-0.8660254037844393,0.5)-- (-0.8660254037844384,3.2320508075688785);
\draw [line width=1.4pt,dotted] (1.8660254037844386,0.5)-- (1.8660254037844393,3.232050807568877);
\draw [shift={(0.5,1.8660254037844397)},line width=1.4pt]  plot[domain=1.308996938995747:1.832595714594046,variable=\t]({1*1.9318516525781362*cos(\t r)+0*1.9318516525781362*sin(\t r)},{0*1.9318516525781362*cos(\t r)+1*1.9318516525781362*sin(\t r)});
\draw [shift={(0.5,1.8660254037844397)},line width=1.4pt]  plot[domain=2.356194490192345:2.879793265790644,variable=\t]({1*1.9318516525781366*cos(\t r)+0*1.9318516525781366*sin(\t r)},{0*1.9318516525781366*cos(\t r)+1*1.9318516525781366*sin(\t r)});
\draw [shift={(0.5,1.8660254037844397)},line width=1.4pt]  plot[domain=3.4033920413889422:3.9269908169872414,variable=\t]({1*1.931851652578137*cos(\t r)+0*1.931851652578137*sin(\t r)},{0*1.931851652578137*cos(\t r)+1*1.931851652578137*sin(\t r)});
\draw [shift={(0.5,1.8660254037844397)},line width=1.4pt]  plot[domain=4.4505895925855405:4.974188368183839,variable=\t]({1*1.9318516525781375*cos(\t r)+0*1.9318516525781375*sin(\t r)},{0*1.9318516525781375*cos(\t r)+1*1.9318516525781375*sin(\t r)});
\draw [shift={(0.5,1.8660254037844397)},line width=1.4pt]  plot[domain=5.497787143782138:6.021385919380436,variable=\t]({1*1.9318516525781373*cos(\t r)+0*1.9318516525781373*sin(\t r)},{0*1.9318516525781373*cos(\t r)+1*1.9318516525781373*sin(\t r)});
\draw [shift={(0.5,1.8660254037844397)},line width=1.4pt]  plot[domain=0.26179938779914874:0.7853981633974476,variable=\t]({1*1.9318516525781362*cos(\t r)+0*1.9318516525781362*sin(\t r)},{0*1.9318516525781362*cos(\t r)+1*1.9318516525781362*sin(\t r)});
\begin{scriptsize}
\draw [fill=black] (0,0) circle (1.4pt);
\draw[color=black] (0.01173979493001023,-0.120104414450016865) node {$v_2$};
\draw [fill=black] (1,0) circle (1.4pt);
\draw[color=black] (1.0563041587443374,-0.104455510123210493) node {$v_3$};
\draw [fill=black] (1.8660254037844386,0.5) circle (1.4pt);
\draw[color=black] (1.940467253208899,0.4080461065796983) node {$v_4$};
\draw [fill=black] (2.366025403784439,1.3660254037844388) circle (1.4pt);
\draw[color=black] (2.4799343830130913,1.454804818351484) node {$v_5$};
\draw [fill=black] (2.366025403784439,2.3660254037844384) circle (1.4pt);
\draw[color=black] (2.4799343830130913,2.4524224691853904) node {$v_6$};
\draw [fill=black] (1.8660254037844393,3.232050807568877) circle (1.4pt);
\draw[color=black] (1.909169444555286,3.320936659323144) node {$v_7$};
\draw [fill=black] (1,3.7320508075688776) circle (1.4pt);
\draw[color=black] (1.0563041587443374,3.8295260499443513) node {$v_8$};
\draw [fill=black] (0,3.7320508075688776) circle (1.4pt);
\draw[color=black] (-0.05085582237721537,3.8295260499443513) node {$v_9$};
\draw [fill=black] (-0.8660254037844384,3.2320508075688785) circle (1.4pt);
\draw[color=black] (-0.9232822385966721,3.3483222418950556) node {$v_{10}$};
\draw [fill=black] (-1.3660254037844388,2.3660254037844393) circle (1.4pt);
\draw[color=black] (-1.4875205335446863,2.468071373512197) node {$v_{11}$};
\draw [fill=black] (-1.3660254037844393,1.3660254037844402) circle (1.4pt);
\draw[color=black] (-1.4875205335446863,1.4626292705148871) node {$v_0$};
\draw [fill=black] (-0.8660254037844393,0.5) circle (1.4pt);
\draw[color=black] (-0.8919844299430594,0.4080461065796983) node {$v_1$};
\end{scriptsize}
\end{tikzpicture}
  \caption{The constructed tour for $s=2$. The black and red edges have cost 1 and 2, respectively. $T$ is the straight tour, the optimal tour $T^*$ is dotted.}
  \label{12TSP 3 Opt lower}
\end{figure}

From the construction it is easy to see that the cost of the tour $T$ is $8s$ while that of the optimal tour is $6s$. 

\begin{lemma} \label{3Opt lower optimal}
For $I_s$ with $s\geq 6$ the tour $T$ constructed above is 3-Opt++-optimal .
\end{lemma}

\begin{proof}
Since the constructed instances $(I_s)_{n\in \N}$ do not contain 1-paths of length 0, the tour $T$ is 3-Opt++-optimal if and only if $T$ is 3-optimal. By construction, the constructed family of instance $(I_s)_{s\in \N}$ is regular. Hence, by Lemma \ref{regular k optimality} it is enough to check that the constructed tour for $I_{6}$ is 3-optimal. We checked this using a self-written computer program that generated  all possible 3-moves for the tour and observed that none of them is improving. 
\end{proof}

\begin{theorem}
The approximation ratio of the 3-Opt++ algorithm for \textsc{(1,2)-TSP} is at least $\frac{4}{3}$.
\end{theorem}

\begin{proof}
By Lemma \ref{3Opt lower optimal}, the tour $T$ for $I_s$ with $s\geq 6$ is 3-Opt++-optimal. Moreover, we know that the length of $T$ is $8s$ and that of the optimal tour $T^*$ is $6s$. Therefore, we get an approximation ratio of $\frac{8s}{6s}=\frac{4}{3}$.
\end{proof}

\subsection{Upper Bound on the Approximation Ratio of the 3-Opt++ Algorithm}
For the upper bound assume that an instance with a 3-Opt++-optimal tour $T$ is given. Let the counters be distributed as in Definition \ref{counter distribution}. Since every 3-Opt++-optimal tour is also 3-optimal, we can apply the results from the upper bound of the 3-Opt algorithm. 

\begin{lemma} \label{3Opt Single}
Every 1-path in a 3-Opt++-optimal tour $T$ having a good counter has exactly two edges. 
\end{lemma}

\begin{proof}
Let the good counter be assigned to the vertex $w$ from the 1-path of length 0 consisting of the vertex $u$ which is incident to the edges $\{v_1,u\}$ and $\{u,v_2\}$ of cost 2 in $T$. Then $\{u,w\}$ has cost 1 and assume that $w$ is not the internal vertex of a 1-path of length 2. By Lemma \ref{3Opt 1-edge connecting endpoints} $w$ cannot be an endpoint of a 1-path. Hence, there is an edge $\{p,w\}\in T$, $c(p,w)=1$ such that $p$ is not endpoint of a 1-path. W.l.o.g.\ assume that $p$ lies on the same side of $\{u,w\}$ as $v_1$ (Figure \ref{3-Opt++ Moves}). Now, we can replace the edges $\{u,v_2\}$ and $\{p,w\}$ by $\{v_2,p\}$ and $\{u,w\}$. 
As 
\begin{align*}
c(u,v_2)+ c(p,w)= 2+1\geq c(v_2,p)+c(u,w),
\end{align*}
after the 2-move the cost of $T$ does not increase. Moreover, we have one 1-path of length 0 less than before, namely $u$. Since $p$ is not endpoint of a 1-path we did not create new 1-paths of length 0, contradicting the 3-Opt++-optimality of $T$.
\end{proof}

\begin{lemma} \label{3Opt 1-paths points}
Every 1-path with $x$ edges in a 3-Opt++-optimal tour $T$ has at most $2x$ counters.
\end{lemma}

\begin{proof}
For $x=0$ the 1-path consists of one vertex and by Corollary \ref{3 Opt at most one counter endpoint} it cannot have any counters. For $x\geq 1$ every 1-path with $x$ edges has two endpoints and $x-1$ internal vertices. If the 1-path does not have a good counter, then each of the inner vertices has two incident edges in the optimal tour and can have at most two bad counters. By Corollary \ref{3 Opt at most one counter endpoint} the two endpoints can each have at most one counter. Therefore, the 1-path can have in total at most $2x$ counters. In the case that the 1-path does have a good counter by Lemma \ref{3Opt Single} the 1-path must have exactly two edges and there are at most $2\cdot 2=4$ good counters on the internal vertex. Moreover, by Lemma \ref{3 Opt endpoints no counter} the endpoints of this 1-path cannot have any counters. Thus, this 1-path can have at most $4$ counters.
\end{proof}

\begin{theorem}
The approximation ratio of the 3-Opt++ algorithm for \textsc{(1,2)-TSP} is at most $\frac{4}{3}$.
\end{theorem}

\begin{proof}
By Lemma \ref{3Opt 1-paths points}, we distributed at most $2h$ counters where $h$ is the number of edges with cost 1 in $T$. Hence, by Lemma \ref{3 Opt ratio bound} the approximation ratio is at most $1+\frac{2}{4+2}=\frac{4}{3}$.
\end{proof}

\bibliographystyle{plain}

\end{document}